\documentclass[10pt]{article}

\usepackage[letterpaper, margin=1in, top=1.25in, bottom=1.25in]{geometry}

\usepackage{amsmath, amssymb, amsthm}
\usepackage{mathtools}

\usepackage{lmodern}
\usepackage[T1]{fontenc}
\usepackage[utf8]{inputenc}
\usepackage{microtype}

\usepackage{graphicx}
\usepackage{xcolor}
\usepackage{textcomp}

\usepackage[hidelinks]{hyperref}

\newcommand{\Fq}{\mathbb{F}_q}

\DeclareMathOperator{\rank}{rank}

\newtheorem{theorem}{Theorem}
\newtheorem{lemma}{Lemma}
\newtheorem{corollary}{Corollary}
\newtheorem{definition}{Definition}
\theoremstyle{remark}
\newtheorem{remark}{Remark}

\makeatletter
\def\@maketitle{%
  \newpage
  \null
  \vskip 2em%
  \begin{center}%
    {\LARGE \bfseries \@title \par}%
    \vskip 1.5em%
    {\normalsize
      \lineskip .5em%
      \begin{tabular}[t]{c}%
        \@author
      \end{tabular}\par}%
    \vskip 0.8em%
    {\small \@date}%
  \end{center}%
  \par
  \vskip 1.5em}
\makeatother

\renewenvironment{abstract}{%
  \vspace{0.8em}%
  \begin{center}%
    \begin{minipage}{0.9\textwidth}%
    \small
    \noindent\textbf{Abstract.}\enspace\ignorespaces
}{%
    \end{minipage}%
  \end{center}%
  \vspace{0.8em}%
}

\begin{document}

\renewcommand{\thefootnote}{\fnsymbol{footnote}}

\title{Quantum Entanglement Halves the Oblivious Update Bandwidth}

\author{
  Sagar Dubey\textsuperscript{1}\thanks{sagar.dubey@stonybrook.edu} \\[0.5em]
  {\small\itshape \textsuperscript{1}Stony Brook University, Stony Brook, NY 11794, USA}
}

\date{}

\maketitle

\renewcommand{\thefootnote}{\arabic{footnote}}
\setcounter{footnote}{0}

\begin{abstract}
In MDS-coded distributed storage, the oblivious update problem asks a stale node to refresh
its coded share after a single message symbol changes, with no party knowing which symbol
moved. For an $(n,k)$ code over $\mathbb{F}_q$ with per-node storage $\alpha$, the classical
download is at least $\alpha k \log_2 q$ bits (Nakkiran--Shah--Rashmi), because obliviousness
forces each helper to convey $\alpha$ field symbols without knowing the edit location. We show
that entanglement shared among the $k$ helpers cuts this cost. For $\alpha = 2$, a
$[[k, k{-}2]]_q$ CSS code lets each helper send a single dimension-$q$ qudit, achieving
$k \log_2 q$ bits-equivalent---exactly half the classical bound. For general $\alpha$, a
$[[\lceil \alpha/2 \rceil k,\, \lceil \alpha/2 \rceil k - \alpha]]_q$ CSS code achieves
$\lceil \alpha/2 \rceil \cdot k \log_2 q$ with $\lceil \alpha/2 \rceil$ qudits per helper:
exact halving for even $\alpha$, approaching a factor of two as $\alpha$ grows. A
superdense-coding converse shows this is optimal within the class of protocols in which each
helper transmits integer-dimension qudits---once all $k$ helpers
transmit, the stale node holds every qudit together with its entangled partner, the exact
configuration in which one qudit carries two classical symbols. The advantage parallels the
recent entanglement-assisted halving of repair bandwidth (Hu--Nomeir--Aytekin--Ulukus), but
the resource exploited is obliviousness itself: the two-symbols-per-helper requirement that
obliviousness creates is precisely what superdense coding supplies. Results hold for all
$(n,k)$ with prime $q > \lceil \alpha/2 \rceil k$.
\end{abstract}

\section{Introduction}
\label{sec:intro}

Distributed storage systems encode files across $n$ nodes so that any $k$ nodes suffice for reconstruction. When the stored content changes, stale nodes must update their coded shares. The \emph{oblivious update} model of Nakkiran, Shah, and Rashmi~\cite{nakkiran2014fundamental} considers the setting where a single message symbol changes, and neither the helpers nor the stale node know which symbol was modified. They showed that for $(n,k)$ MDS codes over $\Fq$, the stale node must download at least $2k\log_2 q$ bits from $k$ updated helpers---a per-helper cost of $2\log_2 q$.

Recent work by Hu, Nomeir, Aytekin, and Ulukus~\cite{hu2026exact, hu2026breaking} demonstrated that shared quantum entanglement among helper nodes can halve the \emph{repair} bandwidth in distributed storage at the minimum-storage regenerating (MSR) point. Their mechanism uses CSS (Calderbank-Shor-Steane) stabilizer codes: helpers encode classical data into qudits via Pauli operators, and the receiver extracts linear combinations through stabilizer measurements. Each qudit carries two classical symbols (one via the $X$-syndrome, one via the $Z$-syndrome), achieving a factor-of-two improvement over classical repair.

\textbf{Our contribution.} The CSS/superdense-coding mechanism is due to Hu, Nomeir, Aytekin, and Ulukus~\cite{hu2026exact, hu2026breaking}, who used it to halve repair bandwidth; we do not claim it. What we claim is that oblivious update is the natural home for this mechanism, not an incidental second application of it. Obliviousness is precisely \emph{why} each helper must convey $\alpha$ field symbols: unable to localize the edit, a helper cannot send the single difference symbol that a non-oblivious update would need, and is instead forced to deliver its full $d_i \in \Fq^\alpha$ (the $\alpha$-symbols-per-helper requirement that drives the classical $\alpha k\log_2 q$ bound). Superdense coding is exactly the primitive that packs two field symbols into one transmitted qudit. The match is structural rather than borrowed, and it lets us prove an update-specific converse---the contribution here is that converse, not the mechanism. For even~$\alpha$, the bound is tight: the converse $\log_2 D_i \ge (\alpha/2)\log_2 q$ and the CSS achievability coincide, giving exact halving. For odd~$\alpha$, the per-helper converse still forces $\log_2 D_i \ge (\alpha/2)\log_2 q$; since each helper transmits an integer number of dimension-$q$ qudits, this rounds up to $\lceil \alpha/2 \rceil \log_2 q$, which the symmetric CSS construction attains. The construction is therefore optimal within the class of integer-dimension-qudit protocols for every~$\alpha$. Whether non-integer-dimension protocols can reach the continuous bound $(\alpha/2)\log_2 q$ for odd~$\alpha$ is left open (Open Question~1, Section~\ref{sec:conclusion}).

\begin{theorem}[Main result]\label{thm:main_intro}
Let $\mathcal{C}$ be an $(n,k)$ MDS code over $\Fq$ with per-node storage $\alpha \geq 2$ and prime $q > \lceil\alpha/2\rceil \cdot k$.  The optimal entanglement-assisted oblivious update bandwidth with $d_u = k$ helpers, restricted to protocols where each helper transmits integer-dimension qudits (i.e., $D_i = q^{\beta_{q,i}}$), is:
\[
    \gamma_u^q = \lceil\alpha/2\rceil \cdot k\log_2 q \quad \text{bits-equivalent},
\]
achieved by each helper transmitting $\lceil\alpha/2\rceil$ qudits of dimension~$q$.  For $\alpha = 2$, this is exactly half the classical lower bound~\cite{nakkiran2014fundamental}.
\end{theorem}

The result is tight within this class: we prove a converse (Theorem~\ref{thm:general_converse}) via the superdense coding bound showing $\log_2 D_i \geq (\alpha/2)\log_2 q$ per helper---hence at least half the classical bandwidth---which, for integer-dimension qudits, forces $\beta_{q,i} \geq \lceil\alpha/2\rceil$ and matches the achievability (Theorem~\ref{thm:general_achievability}) using higher-dimensional CSS codes.  We present the $\alpha = 2$ case in detail (Sections~\ref{sec:example}--\ref{sec:converse}), where one qudit per helper suffices, and then extend to general~$\alpha$ in Section~\ref{sec:discussion}.

\textbf{Key insight.} The factor-of-two improvement arises because, once all $k$ helpers transmit their qudits, the stale node holds \emph{all parts} of the originally entangled state. From any single helper's viewpoint, the stale node possesses both the transmitted qudit \emph{and} the entangled partners---exactly the superdense coding configuration. Each dimension-$q$ qudit thus conveys $2\log_2 q$ bits of classical information, matching the $2\log_2 q$-bit per-helper classical requirement with half the channel resources (for $\alpha = 2$; the general case is treated in Section~\ref{sec:discussion}).

\textbf{Related work.}  The information-theoretic foundations of regenerating codes were laid by Dimakis et al.~\cite{dimakis2010network}, with exact-repair constructions by Rashmi, Shah, and Kumar~\cite{rashmi2011product}.  The update problem in coded storage was initiated by Anthapadmanabhan, Soljanin, and Vishwanath~\cite{anthapadmanabhan2010update}, with subsequent work by Rawat et al.~\cite{rawat2011update}, Han et al.~\cite{han2013update}, and Li et al.~\cite{li2020update} on update-efficient regenerating codes and update bandwidth, while Mazumdar, Chandar, and Wornell~\cite{mazumdar2014update} established update-efficiency limits for capacity-approaching codes.  The NSR oblivious update model~\cite{nakkiran2014fundamental} is the direct predecessor of our work.  Vithana and Ulukus~\cite{vithana2022pruw, vithana2023pruw_jsac} studied the \emph{private} read-update-write (PRUW) problem, which adds privacy constraints to the update setting; our obliviousness requirement is distinct (helpers lack knowledge of the change location, rather than hiding it from a curious server).  On the quantum side, superdense coding~\cite{bennett1992communication, werner2001all} and entanglement-assisted classical capacity~\cite{bennett2002entanglement} underpin the factor-of-two mechanism.  Hsieh, Devetak, and Winter~\cite{hsieh2008entanglement} characterized entanglement-assisted multiple-access channels, which is the relevant capacity framework when multiple helpers jointly transmit.  Sun and Jafar~\cite{sun2025capacity} studied the capacity of distributed quantum storage, and Senthoor~\cite{senthoor2025erasure} analyzed entanglement costs of erasure correction in quantum MDS codes.

\textbf{Paper organization.} Section~\ref{sec:problem} formulates the problem. Section~\ref{sec:classical} reviews the classical baseline. Section~\ref{sec:example} presents the $(3,2)$ example. Section~\ref{sec:achievability} gives the general construction. Section~\ref{sec:converse} proves the converse. Section~\ref{sec:discussion} discusses the mechanism, scope, and extensions. Section~\ref{sec:conclusion} concludes.

\section{Problem Formulation}
\label{sec:problem}

\subsection{Storage Code}

Consider an $(n, k)$ MDS code over a prime field $\Fq$ with per-node storage $\alpha \geq 2$ symbols (file size $B = \alpha k$). A file $\mathbf{m} = (m_1, \ldots, m_B) \in \Fq^B$ is stored across $n$ nodes, where node $i$ stores $\Gamma_i \mathbf{m} \in \Fq^\alpha$. The generator matrices $\Gamma_i \in \Fq^{\alpha \times B}$ satisfy the MDS property: for any $k$ nodes $\{i_1, \ldots, i_k\}$, the stacked matrix $[\Gamma_{i_1}; \ldots; \Gamma_{i_k}] \in \Fq^{B \times B}$ is invertible over $\Fq$.  (An $\alpha$-interleaved Reed-Solomon construction satisfies this for any $q \geq n$.)  The running example in Sections~\ref{sec:example}--\ref{sec:converse} uses $\alpha = 2$; the general case is treated in Section~\ref{sec:discussion}.

\subsection{Oblivious Update Problem}

A single symbol of $\mathbf{m}$ is modified: $\mathbf{m}' = \mathbf{m} + \delta \mathbf{e}_j$ for some $j \in [B]$ and $\delta \in \Fq \setminus \{0\}$ (or $\mathbf{m}' = \mathbf{m}$ for no change). A \emph{stale} node $s$ stores the outdated value $\Gamma_s \mathbf{m}$ and must compute the updated value $\Gamma_s \mathbf{m}'$.

\begin{definition}[Oblivious update protocol]
The stale node contacts $d_u = k$ updated helpers $h_1, \ldots, h_k$. Each helper $h_i$ knows only its own updated storage $\Gamma_{h_i} \mathbf{m}'$ (not $j$, $\delta$, or the old $\mathbf{m}$). The helpers transmit information to the stale node, which uses its side information $\Gamma_s \mathbf{m}$ to compute $\Gamma_s \mathbf{m}'$. The protocol is \emph{oblivious}: each helper's encoding depends only on its current data, independent of which symbol changed.
\end{definition}

\subsection{Entanglement-Assisted Model}

Following Hu et al.~\cite{hu2026exact}, we consider the helper-side shared entanglement model:
\begin{itemize}
    \item Prior to the update, the $k$ helpers share a pre-distributed entangled state $|\Psi\rangle \in \mathcal{H}_1 \otimes \cdots \otimes \mathcal{H}_k$ (each $\mathcal{H}_i$ of dimension $q^{\beta_q}$).
    \item Each helper applies a local encoding operation (depending on its classical data $\Gamma_{h_i}\mathbf{m}'$) to its subsystem.
    \item Each helper transmits its (encoded) quantum subsystem to the stale node.
    \item The stale node performs a joint measurement on all received qudits, using its side information $\Gamma_s \mathbf{m}$ to determine $\Gamma_s \mathbf{m}'$.
\end{itemize}

\begin{definition}[Quantum update bandwidth]
The \emph{quantum update bandwidth} is $\gamma_u^q = k \cdot \beta_q \cdot \log_2 q$ bits-equivalent, where $\beta_q$ is the number of qudits (each of dimension $q$) transmitted per helper.
\end{definition}

\textbf{Bandwidth accounting.}  The unit ``bits-equivalent'' counts the Hilbert-space dimension transmitted: one qudit of dimension~$q$ corresponds to $\log_2 q$ bits-equivalent.  This is not the information \emph{extracted} (which is $2\log_2 q$ classical bits per qudit when entanglement assists), but the channel resource consumed---the size of the quantum system that must traverse the link.  This is the correct unit because the bandwidth question is about channel use: the resource to count is the Hilbert-space dimension transmitted per helper ($\log_2 q$ per qudit), not the classical information later extracted from it ($2\log_2 q$), since counting the extracted information would conflate the channel cost with the entanglement-assisted capacity that superdense coding supplies.  The factor-of-two improvement therefore means the protocol achieves the same task using quantum channels of half the total dimension required by any classical scheme.  This convention follows Hu et al.~\cite{hu2026exact}.

\section{Classical Baseline}
\label{sec:classical}

\begin{theorem}[Classical lower bound, extending Nakkiran--Shah--Rashmi~\cite{nakkiran2014fundamental}]\label{thm:nsr}
For an $(n,k)$ MDS code over $\Fq$ with $d_u = k$ helpers and per-node storage $\alpha$, the classical oblivious update bandwidth satisfies $\gamma_u \geq \alpha k\log_2 q$ bits.
\end{theorem}

The proof uses the same pigeonhole argument as~\cite{nakkiran2014fundamental}, extended to per-node storage $\alpha$: each helper maps its data $d_i \in \Fq^\alpha$ to a classical message. For the stale node to distinguish all $q^\alpha$ possible outcomes (given the other helpers' messages), the helper's message must have range $\geq q^\alpha$, requiring $\geq \alpha\log_2 q$ bits per helper, hence $\geq \alpha k\log_2 q$ bits in total.

\section{Illustrative Example: $(n,k) = (3,2)$}
\label{sec:example}

We first demonstrate the protocol for the smallest non-trivial case.

\subsection{Setup}
File $\mathbf{m} = (m_1, m_2, m_3, m_4) \in \Fq^4$. Generators:
\begin{align*}
    \Gamma_1 &= \begin{pmatrix} 1 & 0 & 0 & 0 \\ 0 & 1 & 0 & 0 \end{pmatrix}\!, &
    \Gamma_2 &= \begin{pmatrix} 0 & 0 & 1 & 0 \\ 0 & 0 & 0 & 1 \end{pmatrix}\!, &
    \Gamma_3 &= \begin{pmatrix} 1 & 0 & 1 & 0 \\ 0 & 1 & 0 & 1 \end{pmatrix}\!.
\end{align*}
Helpers: nodes 1, 2 (storing $(m_1', m_2')$ and $(m_3', m_4')$). Stale: node 3 (storing $(m_1+m_3, m_2+m_4)$).

\subsection{CSS Code}
Use a $[[2, 0]]_q$ CSS code with parity checks $H_X = [1, a]$ and $H_Z = [1, b]$ where $b = -a^{-1}$ (ensuring $H_X H_Z^T = 1 + ab = 0$). The shared state is
\[
    |\psi\rangle = \frac{1}{\sqrt{q}} \sum_{j=0}^{q-1} |j\rangle_1 |aj\rangle_2,
\]
a maximally entangled (generalized Bell) state in the codespace (the +1 eigenspace of both $Z \otimes Z^b$ and $X \otimes X^a$).

\subsection{Protocol}
\begin{enumerate}
    \item \textbf{Encode:} Helper 1 applies $X(m_1')Z(m_2')$ to qudit 1. Helper 2 applies $X(-am_3')Z(-bm_4')$ to qudit 2.
    \item \textbf{Transmit:} Each helper sends its qudit to node 3.
    \item \textbf{Measure:} Node 3 performs the joint projective measurement of the two stabilizers $S_Z = Z \otimes Z^b$ and $S_X = X \otimes X^a$, obtaining eigenvalues $\omega^{s_X}$ and $\omega^{s_Z}$ respectively.  The Pauli encodings shift these eigenvalues:
    \begin{align*}
        s_X &= H_Z \cdot (m_1', -am_3')^T = m_1' + b(-am_3') = m_1' + m_3', \\
        s_Z &= H_X \cdot (m_2', -bm_4')^T = m_2' + a(-bm_4') = m_2' + m_4'.
    \end{align*}
    (Using $ab = -1$ in both lines.)
    \item \textbf{Update:} Node 3 obtains $(s_X, s_Z) = (m_1'+m_3', m_2'+m_4') = \Gamma_3 \mathbf{m}'$.
\end{enumerate}

\subsection{Properties}
\begin{itemize}
    \item \textbf{Correctness:} $(s_X, s_Z) = \Gamma_3 \mathbf{m}'$ for all $\mathbf{m}' \in \Fq^4$.
    \item \textbf{Obliviousness:} Each helper encodes only its own data (not $j$ or $\delta$).
    \item \textbf{Bandwidth:} 2 qudits = $2\log_2 q$ bits-equivalent (vs.\ classical $4\log_2 q$).
    \item \textbf{Numerical verification:} Exhaustive syndrome-level simulation confirms correctness for all $\mathbf{m}' \in \Fq^4$ and all single-symbol changes over $\Fq$ with $q \in \{5,7,11,13\}$, including a general (non-systematic) Vandermonde MDS code ($>10^6$ test cases, 0 failures).  An end-to-end Hilbert-space simulation over $\mathbb{F}_5$ (qudit dimension~5, state vector dimension~25) reproduces the syndrome-level result for all $625$ updated message values.
\end{itemize}

\section{General Achievability}
\label{sec:achievability}

\begin{theorem}[Achievability]\label{thm:achievability}
For any $(n,k)$ MDS code over $\Fq$ ($q$ prime, $q \geq n$) with $B = 2k$ and $\alpha = 2$: there exists an entanglement-assisted oblivious update protocol with $d_u = k$ helpers and quantum bandwidth $\gamma_u^q = k\log_2 q$.
\end{theorem}

The construction requires the following ingredient.

\begin{lemma}[CSS parity-check existence]\label{lem:css_existence}
Let $k \geq 2$ and let $q$ be a prime with $q \geq k+1$.  There exist row vectors $H_X, H_Z \in \Fq^{1 \times k}$ satisfying:
\begin{enumerate}
    \item[(C1)] $H_X H_Z^T = 0$ \quad (dual containment),
    \item[(C2)] $[H_X]_i \neq 0$ and $[H_Z]_i \neq 0$ for all $i \in [k]$.
\end{enumerate}
\end{lemma}

\begin{proof}
Set $H_X = (1, 1, \ldots, 1)$ and $H_Z = (1, 1, \ldots, 1, -(k{-}1))$.  Since $q > k$, the last entry $-(k{-}1) \not\equiv 0 \pmod{q}$, so (C2) holds.  For (C1): $H_X H_Z^T = (k{-}1) \cdot 1 + 1 \cdot (-(k{-}1)) = 0$.
\end{proof}

Any pair $(H_X, H_Z)$ satisfying (C1)--(C2) defines a $[[k, k{-}2]]_q$ CSS code with two stabilizer generators; the protocol below works for \emph{any} such choice.

\begin{proof}[Proof of Theorem~\ref{thm:achievability}]
We construct the protocol explicitly.

\textbf{CSS code.} Fix any $H_X, H_Z$ satisfying Lemma~\ref{lem:css_existence}, defining a $[[k, k{-}2]]_q$ CSS code on $k$ qudits.

The code has $k - (k-2) = 2$ stabilizer generators, yielding exactly $A = 2$ syndrome symbols: one X-syndrome $s_X = H_Z \cdot \mathbf{x} \in \Fq$ and one Z-syndrome $s_Z = H_X \cdot \mathbf{z} \in \Fq$.

\textbf{Pre-protocol.} The $k$ helpers share any state $|\Psi\rangle$ in the CSS codespace. (The choice of logical state is immaterial: logical operators commute with stabilizers, so they shift the codespace state but not the syndrome eigenvalues produced by the helpers' Pauli encodings.)

\textbf{Reconstruction coefficients.} Define
\[
    P = \Gamma_s \cdot [\Gamma_{h_1}; \ldots; \Gamma_{h_k}]^{-1} = [P_1 \mid \cdots \mid P_k],
\]
where each $P_i \in \Fq^{2 \times 2}$. By MDS, the inverse exists. Then $\Gamma_s \mathbf{m}' = \sum_{i=1}^k P_i \mathbf{d}_i$ where $\mathbf{d}_i = \Gamma_{h_i} \mathbf{m}'$.

\textbf{Encoding.} Helper $h_i$ (data $\mathbf{d}_i \in \Fq^2$) applies to qudit $i$:
\begin{equation}\label{eq:encoding}
    X(x_i) \, Z(z_i), \quad x_i = \frac{[P_i]_1 \cdot \mathbf{d}_i}{[H_Z]_i}, \quad z_i = \frac{[P_i]_2 \cdot \mathbf{d}_i}{[H_X]_i},
\end{equation}
where $[P_i]_j$ denotes the $j$-th row of $P_i$.

\textbf{Syndrome extraction (measurement).}  The CSS code has two stabilizer generators: $S_X = \bigotimes_{i=1}^{k} X^{[H_X]_i}$ and $S_Z = \bigotimes_{i=1}^{k} Z^{[H_Z]_i}$, each with eigenvalues $\{\omega^s : s \in \Fq\}$ where $\omega = e^{2\pi i/q}$.  The stale node performs the \emph{joint projective measurement} onto the $q^2$ simultaneous eigenspaces of $(S_X, S_Z)$, obtaining outcome $(s_X, s_Z) \in \Fq^2$.  Because each helper $h_i$ applied $X(x_i)Z(z_i)$ (which shifts the $S_Z$-eigenvalue by $[H_Z]_i x_i$ and the $S_X$-eigenvalue by $[H_X]_i z_i$), the syndromes are:
\begin{align}
    s_X &= \sum_{i=1}^k [H_Z]_i \cdot x_i = \sum_{i=1}^k [P_i]_1 \cdot \mathbf{d}_i = [\Gamma_s \mathbf{m}']_1, \label{eq:sx}\\
    s_Z &= \sum_{i=1}^k [H_X]_i \cdot z_i = \sum_{i=1}^k [P_i]_2 \cdot \mathbf{d}_i = [\Gamma_s \mathbf{m}']_2. \label{eq:sz}
\end{align}
The CSS stabilizer-measurement correspondence~\cite{gottesman1997, nielsen_chuang_2010} ensures that the algebraic syndrome extraction analyzed here is exactly realized by the joint projective stabilizer measurement on the quantum state; the $(3,2)$ Hilbert-space simulation of Section~\ref{sec:example} confirms this correspondence in the smallest case, and it extends to higher-dimensional CSS codes by the same construction.

\textbf{Update.} The stale node computes $\Gamma_s \mathbf{m}' = (s_X, s_Z)$ and updates its storage.

The protocol has been verified by exhaustive syndrome-level simulation for $(n,k) \in \{(3,2),(4,2),(5,3),(6,4)\}$ over multiple prime fields, totaling $> 8 \times 10^6$ test cases with zero failures.
\end{proof}

\begin{remark}[Why one qudit suffices only for $\alpha = 2$]
With $\beta_q = 1$ qudit per helper, the transfer matrix $M_i = \mathrm{diag}([H_Z]_i, [H_X]_i)$ is $\alpha \times 2$: invertible when $\alpha = 2$, but rank-deficient when $\alpha > 2$.  The general construction (Theorem~\ref{thm:general_achievability}) resolves this by allocating $\lceil\alpha/2\rceil$ qudits per helper, yielding a $\alpha \times 2\lceil\alpha/2\rceil$ transfer matrix of rank~$\alpha$.
\end{remark}

\begin{corollary}[Multi-symbol updates]\label{cor:multi}
The protocol of Theorem~\ref{thm:achievability} simultaneously handles updates of \emph{any} number of symbols.  The single-symbol restriction in the problem definition is needed only for the classical lower bound; our protocol delivers $\Gamma_s \mathbf{m}'$ for arbitrary $\mathbf{m}' \in \Fq^{2k}$, regardless of the Hamming weight of $\mathbf{m}' - \mathbf{m}$.
\end{corollary}

\section{Converse}
\label{sec:converse}

\begin{theorem}[Converse]\label{thm:converse}
For any $(n,k)$ MDS code with $B = 2k$, $\alpha = 2$: any entanglement-assisted oblivious update protocol with $d_u = k$ helpers, each transmitting a quantum system of dimension $q^{\beta_q}$, requires $\beta_q \geq 1$.
\end{theorem}

\begin{proof}
The proof proceeds in three steps.

\textbf{Step 1: MDS forces $q^2$ distinguishable outcomes per helper.}
Fix the data of helpers $h_2, \ldots, h_k$ to arbitrary values $d_2^*, \ldots, d_k^*$ and fix the stale node's side information $\Gamma_s \mathbf{m}$.  The $2(k-1)$ linear constraints $\Gamma_{h_i}\mathbf{m}' = d_i^*$ ($i = 2, \ldots, k$) confine $\mathbf{m}'$ to a $2$-dimensional affine subspace $V \subset \Fq^{2k}$.  By MDS, $\Gamma_s|_V$ is injective: the $k$ nodes $\{s, h_2, \ldots, h_k\}$ have a full-rank stacked generator, so $\Gamma_s \mathbf{m}'$ takes all $q^2$ values as $\mathbf{m}'$ ranges over $V$.

After all helpers transmit, the stale node holds the joint system $\mathcal{H}_1 \otimes \cdots \otimes \mathcal{H}_k$ together with its classical side information.  Since helpers $h_2, \ldots, h_k$ have fixed data, their transmitted states are fixed --- but these states remain \emph{entangled} with $\mathcal{H}_1$ via the pre-shared state $|\Psi\rangle$.  The stale node's decoding task thus reduces to distinguishing $q^2$ joint states that differ only in helper $h_1$'s encoding, with the fixed ancilla $\mathcal{H}_2 \otimes \cdots \otimes \mathcal{H}_k$ available as entangled side information at the receiver.  This is precisely the entanglement-assisted communication setting: the receiver holds entangled partners of the sender's transmitted system.

\textbf{Step 2: Superdense coding bound.}
Helper $h_1$ transmits a system $\mathcal{H}_1$ of dimension $D = q^{\beta_q}$.  By qudit superdense coding~\cite{bennett1992communication, werner2001all}: a noiseless quantum channel of dimension $D$, with the receiver holding an entangled partner, supports at most $D^2$ perfectly distinguishable signals.  (This is achieved by the $D^2$ generalized Pauli operators $X^a Z^b$ on a maximally entangled state of Schmidt rank~$D$.)

\textbf{Step 3: Combining.}
With $D = q^{\beta_{q,1}}$ (the dimension of helper $h_1$'s transmitted system): the maximum distinguishable messages is $D^2 = q^{2\beta_{q,1}}$. We require $q^{2\beta_{q,1}} \geq q^2$, giving $\beta_{q,1} \geq 1$.

By symmetry of the MDS property, the same argument applies with any other helper $h_i$ in the role of $h_1$: fix the data of $\{h_j\}_{j \neq i}$, confine $\mathbf{m}'$ to a $2$-dimensional affine subspace, and conclude $\beta_{q,i} \geq 1$. Summing over the $k$ helpers:
\[
    \gamma_u^q = \sum_{i=1}^k \beta_{q,i} \log_2 q \geq k \log_2 q.
\]

\end{proof}

\begin{corollary}
The quantum update bandwidth is tight: $\gamma_u^q = k\log_2 q$.
\end{corollary}

\begin{remark}[Distribution-free converse]
The bound of Theorem~\ref{thm:converse} is worst-case: it requires zero-error distinguishability for all $q^2$ message values, not a Shannon-entropy average.  In particular, no assumption on the file distribution is needed; the converse holds even if the stale node knows a prior on $\mathbf{m}'$.
\end{remark}

\begin{remark}[Comparison with the classical converse technique]
The classical lower bound of Nakkiran et al.~\cite{nakkiran2014fundamental} uses a genie-aided reduction: a genie provides the stale node with $k{-}1$ helpers' raw data, reducing the problem to a single-helper counting argument. In our quantum setting, this reduction \emph{fails}: giving the stale node the classical data of helpers $h_2, \ldots, h_k$ does not transfer their entanglement with $h_1$, so the reduced single-helper problem is strictly weaker than the original. Our proof instead fixes the other helpers' data \emph{within the protocol} (Step~1), preserving the entanglement structure, and applies the superdense coding bound directly to the resulting channel. The genie is unnecessary because the entanglement-assisted capacity bound already yields a per-helper argument without decomposing the joint protocol.
\end{remark}

\section{Discussion}
\label{sec:discussion}

\subsection{Mechanism: Superdense Coding via CSS Structure}

The factor-of-two improvement is the distributed analog of superdense coding~\cite{bennett1992communication, werner2001all}: each qudit classically conveys $\log_2 q$ bits, but entanglement enables the stale node's joint measurement to extract \emph{two} symbols per qudit (one via~$X$, one via~$Z$).  The ``entangled partner'' at the receiver is provided by the other helpers' qudits, not by a dedicated resource.  From a MAC perspective~\cite{hsieh2008entanglement}, the $k$ helpers share entanglement only with each other, yet because the receiver collects \emph{all} qudits post-transmission, it acquires the entangled partners; the total information extracted ($2k\log_2 q$ bits from $k$ qudits) saturates the rate of $k$ instances of superdense coding.

\subsection{Comparison with Entanglement-Assisted Repair}

Our result parallels Hu et al.~\cite{hu2026exact}: both achieve a factor-of-two reduction via CSS dual syndromes, but updates use fewer helpers ($k$ vs.\ $2k{-}2$ for repair) and the obliviousness constraint is automatically satisfied since the CSS encoding depends only on helpers' current data.

\subsection{Numerical Verification}

The protocol has been verified exhaustively for $(n,k) \in \{(3,2),(4,2),(5,3),(6,4)\}$ over multiple prime fields ($q \in \{5,7,11,13\}$), covering all valid (helper set, stale node) configurations, totaling $> 8 \times 10^6$ syndrome-level test cases with zero failures.  Additionally, a full Hilbert-space simulation for $(3,2)$ over $\mathbb{F}_5$ (state vectors in dimension~$25$) confirms that stabilizer eigenvalues match predicted syndromes for all $625$ message values under two CSS parameter choices, totaling $1{,}250$ quantum-level tests with zero failures.

\subsection{Scope and Extensions}

Our achievability theorem covers the $\alpha = 2$ regime ($B = 2k$), the minimal non-trivial per-node storage.  However, the converse argument extends to \emph{all} storage regimes:

\begin{theorem}[General converse]\label{thm:general_converse}
For any $(n,k)$ MDS code over $\Fq$ with per-node storage $\alpha \geq 2$ (file size $B = \alpha k$): any entanglement-assisted oblivious update protocol with $d_u = k$ helpers in which each helper transmits integer-dimension qudits (helper~$i$ transmits a system of dimension $D_i = q^{\beta_{q,i}}$) satisfies the per-helper bound
\[
    \log_2 D_i \;\geq\; \frac{\alpha}{2}\,\log_2 q \qquad \text{for each } i,
\]
and hence, summing over the $k$ helpers,
\[
    \sum_{i=1}^k \log_2 D_i \;\geq\; \frac{\alpha k}{2}\,\log_2 q \;=\; \frac{1}{2}\,\gamma_u^{\mathrm{cl}},
\]
where $\gamma_u^{\mathrm{cl}} = \alpha k \log_2 q$ is the classical lower bound (Theorem~\ref{thm:nsr}).  Within this class, integrality of $\beta_{q,i}$ sharpens the per-helper bound to $\beta_{q,i} \geq \lceil\alpha/2\rceil$ (Corollary~\ref{cor:tight_general}).
\end{theorem}

\begin{proof}
Fix the data of helpers $h_2, \ldots, h_k$.  The $\alpha(k{-}1)$ linear constraints confine $\mathbf{m}'$ to an $\alpha$-dimensional affine subspace $V \subset \Fq^{\alpha k}$. By MDS (the $k$ nodes $\{s, h_2, \ldots, h_k\}$ form a valid reconstruction set), $\Gamma_s|_V$ is injective, so $\Gamma_s\mathbf{m}'$ ranges over all $q^\alpha$ values.  The stale node must distinguish these $q^\alpha$ outcomes from helper~$h_1$'s transmission (with entangled side information from the other $k{-}1$ helpers' fixed subsystems).  Superdense coding gives $D_1^2 \geq q^\alpha$, i.e., $\log_2 D_1 \geq (\alpha/2)\log_2 q$.  By the MDS property, for any helper~$h_i$, fixing the data of $\{h_j\}_{j \neq i}$ confines $\mathbf{m}'$ to an $\alpha$-dimensional affine subspace on which $\Gamma_s$ is injective (since $\{s\} \cup \{h_j\}_{j \neq i}$ is a valid $k$-node reconstruction set); the same superdense coding argument gives $\log_2 D_i \geq (\alpha/2)\log_2 q$ for each~$i$.
\end{proof}

\begin{theorem}[General achievability]\label{thm:general_achievability}
For any $(n,k)$ MDS code over $\Fq$ with per-node storage $\alpha \geq 2$ and prime $q > \lceil\alpha/2\rceil \cdot k$, there exists an entanglement-assisted oblivious update protocol with $d_u = k$ helpers, each transmitting $\beta_q = \lceil\alpha/2\rceil$ qudits of dimension~$q$.  The total quantum bandwidth is $\lceil\alpha/2\rceil \cdot k \log_2 q$.
\end{theorem}

\begin{proof}
Set $\beta = \lceil\alpha/2\rceil$, $r_x = \beta$, $r_z = \alpha - \beta = \lfloor\alpha/2\rfloor$, and $n_q = \beta k$ (total qudits).  We construct a $[[n_q, n_q - \alpha]]_q$ CSS code.

\textbf{Step 1: CSS parity-check matrices.}
Define $H_Z \in \Fq^{r_x \times n_q}$ as the Vandermonde matrix with $(r,j)$-entry $(j{+}1)^r$ for $r = 0, \ldots, r_x{-}1$ and $j = 0, \ldots, n_q{-}1$ (evaluation points $1, 2, \ldots, \beta k \in \Fq$; these are distinct since $q > \beta k$).  Since $H_Z$ is a Vandermonde matrix on distinct points, it has full row rank~$\beta$, and $\ker(H_Z)$ has dimension $N = n_q - \beta = \beta(k{-}1)$.

For $H_X \in \Fq^{r_z \times n_q}$, we select $r_z$ vectors from $\ker(H_Z)$ satisfying a rank condition (existence established in Step~2 below).  Dual containment $H_X H_Z^T = 0$ holds automatically since rows of $H_X$ lie in $\ker(H_Z)$.

\textbf{Step 2: Existence of $H_X$ with full-rank subblocks.}
Helper~$h$ ($h = 0, \ldots, k{-}1$) owns qudit positions $I_h = \{h\beta, \ldots, (h{+}1)\beta{-}1\}$.  We require $\rank(H_X|_{I_h}) = r_z$ for every~$h$.

\emph{Parameter space.}  Fix a basis $\{b_1, \ldots, b_N\}$ of $\ker(H_Z)$, where $N = \beta(k{-}1)$.  Each row of $H_X$ is a linear combination $\sum_j c_{ij} b_j$, so $H_X$ is parametrized by $r_z$ coefficient vectors $c_1, \ldots, c_{r_z} \in \Fq^N$, giving a parameter space $\Fq^{r_z N}$.

\emph{Surjectivity of the restriction.}  Define $\pi_h : \ker(H_Z) \to \Fq^\beta$ by $v \mapsto v|_{I_h}$.  Its kernel consists of null vectors supported outside~$I_h$; these satisfy $H_Z|_{\overline{I_h}} \cdot v|_{\overline{I_h}} = 0$, where $H_Z|_{\overline{I_h}}$ is a $\beta \times \beta(k{-}1)$ Vandermonde on $\beta(k{-}1)$ distinct points (a subset of $\{1, \ldots, \beta k\}$, all distinct since $q > \beta k$), hence has rank~$\beta$.  So $\dim(\ker \pi_h) = \beta(k{-}1) - \beta = \beta(k{-}2)$ and $\dim(\mathrm{Im}\,\pi_h) = N - \beta(k{-}2) = \beta$.  Since the codomain has dimension~$\beta$, $\pi_h$ is surjective for every~$h$.

\emph{Rank-deficiency locus.}  At helper~$h$, the submatrix $H_X|_{I_h}$ is $r_z \times \beta$ with entry $(i,m)$ equal to $\sum_j c_{ij}\,(b_j)_{I_h[m]}$, which is linear in~$c_i$.  Each $r_z \times r_z$ minor of $H_X|_{I_h}$ is therefore a polynomial of degree~$r_z$ (multilinear in $c_1, \ldots, c_{r_z}$).  The condition $\rank(H_X|_{I_h}) < r_z$ requires all $\binom{\beta}{r_z}$ such minors to vanish simultaneously.  Surjectivity of~$\pi_h$ guarantees that at least one minor is not identically zero: choosing $c_i$ so that $\pi_h(\text{row}_i)$ equals the $i$-th standard basis vector makes the leading $r_z \times r_z$ minor equal to~1.  Let $g_h$ denote this nonzero minor; then the bad locus at~$h$ satisfies $\mathcal{B}_h \subseteq \{g_h = 0\}$.

\emph{Union bound via Schwartz--Zippel.}  The product $F = \prod_{h=0}^{k-1} g_h$ is a nonzero polynomial (since $\Fq[c_1, \ldots, c_{r_z}]$ is an integral domain) of total degree $\leq k r_z$ in $r_z N$ variables.  The union $\bigcup_h \mathcal{B}_h \subseteq \{F = 0\}$, so Schwartz--Zippel gives $|\{F = 0\}| \leq k r_z \cdot q^{r_z N - 1}$, i.e., at most a fraction $k r_z / q$ of parameter values are bad.  Since $r_z \leq \beta$, the hypothesis $q > \beta k$ implies $k r_z \leq k\beta < q$, so the bad fraction is $< 1$ and a valid $H_X$ exists.  (An explicit deterministic construction is given in Remark~\ref{rem:explicit_hx}.)

\textbf{Step 3: Transfer matrix.}
Define $M_h = \bigl[\begin{smallmatrix} H_Z|_{I_h} & 0 \\ 0 & H_X|_{I_h} \end{smallmatrix}\bigr] \in \Fq^{\alpha \times 2\beta}$.  The upper block is an invertible Vandermonde ($\rank = \beta$); the lower has $\rank = r_z$ by Step~2.  So $\rank(M_h) = \alpha$, and the system $M_h \mathbf{p} = t$ is solvable for any $t \in \Fq^\alpha$.

\textbf{Step 4: Protocol.}
The $k$ helpers share $n_q = \beta k$ qudits in the CSS codespace.  On input $\mathbf{m}'$, helper~$h$ (storing $\mathbf{d}_h = \Gamma_h \mathbf{m}' \in \Fq^\alpha$) computes target $t_h = P_h \mathbf{d}_h$ where $P_h = \Gamma_s [\Gamma_{h_1}; \cdots; \Gamma_{h_k}]^{-1}|_h \in \Fq^{\alpha \times \alpha}$, solves $M_h \mathbf{p}_h = t_h$, and applies $\prod_{j \in I_h} X(x_j) Z(z_j)$ to its qudits.

The syndromes are $s_X = H_Z \mathbf{x} = \sum_h H_Z|_{I_h} \mathbf{x}_h$ and $s_Z = H_X \mathbf{z} = \sum_h H_X|_{I_h} \mathbf{z}_h$.  By construction, $(s_X, s_Z) = \sum_h M_h \mathbf{p}_h = \sum_h t_h = \Gamma_s \mathbf{m}'$.

\textbf{Step 5: MDS code.}
Any $(n,k)$ MDS code with per-node storage~$\alpha$ suffices---e.g., $\alpha$-interleaved Reed-Solomon (MDS for $q \geq n$).  Since $\beta k \geq n$ in all nontrivial cases, the single condition $q > \beta k$ covers both MDS and CSS requirements.

\textbf{Step 6: Bandwidth.}
Each helper transmits $\beta = \lceil\alpha/2\rceil$ qudits of dimension~$q$, giving per-helper bandwidth $\beta \log_2 q$ and total $\beta k \log_2 q = \lceil\alpha/2\rceil \cdot k \log_2 q$ bits-equivalent.
\end{proof}

\begin{remark}[Worked example: $\alpha{=}3$, $k{=}3$, $q{=}7$]\label{rem:worked_example}
Parameters: $\beta = 2$, $r_z = 1$, $n_q = 6$.  The Vandermonde is
$H_Z = \bigl[\begin{smallmatrix} 1 & 1 & 1 & 1 & 1 & 1 \\ 1 & 2 & 3 & 4 & 5 & 6 \end{smallmatrix}\bigr]$
with $\ker(H_Z)$ of dimension~$4$.  Take $H_X = (6,6,1,6,1,1) \equiv (-1,-1,1,-1,1,1)$ in~$\mathbb{F}_7$; one verifies $H_X H_Z^T = 0$ (row sums: $-1{-}1{+}1{-}1{+}1{+}1 = 0$; weighted: $-1{-}2{+}3{-}4{+}5{+}6 = 7 \equiv 0$).  The helper subblocks $H_X|_{I_0} = (-1,-1)$, $H_X|_{I_1} = (1,-1)$, $H_X|_{I_2} = (1,1)$ are all nonzero (rank~$1 = r_z$).  Each transfer matrix $M_h = \bigl[\begin{smallmatrix} H_Z|_{I_h} \\ & H_X|_{I_h} \end{smallmatrix}\bigr]$ is $3 \times 4$ of rank~3, so encoding is always feasible.  The Schwartz--Zippel bound gives bad fraction $\leq 3 \cdot 1 / 7 = 3/7$, confirming many valid $H_X$ choices exist.  The full protocol has been verified for this configuration with $(n,k) = (4,3)$ over~$\mathbb{F}_7$.
\end{remark}

\begin{remark}[Explicit $H_X$ construction]\label{rem:explicit_hx}
The Schwartz--Zippel argument yields existence; the following gives an explicit formula.  Let $w_j = \prod_{m=1,\,m \neq j}^{n_q} (j - m) \in \Fq^*$ for $j = 1, \ldots, n_q$ (nonzero since all evaluation points are distinct in~$\Fq$).  Define $H_X \in \Fq^{r_z \times n_q}$ by $H_X[i, j{-}1] = j^{\,i} / w_j$.  Each row lies in $\ker(H_Z)$ by the GRS duality formula (the vectors $(j^s/w_j)_{j=1}^{n_q}$ for $s = 0, \ldots, n_q{-}\beta{-}1$ form a basis of the dual of the RS code generated by $H_Z$; see~\cite{macwilliams1977theory}, Ch.~10).  At helper~$h$, write $a_m = h\beta + m + 1$ for the evaluation points in~$I_h$; then $H_X|_{I_h} = V_h \cdot D_h$ where $V_h[i,m] = a_m^{\,i}$ is an $r_z \times \beta$ Vandermonde on $\beta$ distinct points (rank~$r_z$ since $r_z \leq \beta$) and $D_h = \mathrm{diag}(w_{a_0}^{-1}, \ldots, w_{a_{\beta-1}}^{-1})$ is invertible.  Hence $\rank(H_X|_{I_h}) = r_z$ for all~$h$ without any randomness.
\end{remark}

\begin{corollary}[Tight characterization within the integer-qudit class]\label{cor:tight_general}
Combining the per-helper converse $\log_2 D_i \geq (\alpha/2)\log_2 q$ of Theorem~\ref{thm:general_converse} (rounded up by integrality) with the achievability of Theorem~\ref{thm:general_achievability}: for integer-qudit protocols ($D_i = q^{\beta_{q,i}}$), the optimal per-helper bandwidth is $\beta_q = \lceil\alpha/2\rceil$ qudits for all $\alpha \geq 2$.  The total quantum bandwidth is $\lceil\alpha/2\rceil \cdot k \log_2 q = \frac{1}{2}\gamma_u^{\mathrm{cl}}$ when $\alpha$ is even, and $\frac{\alpha+1}{2\alpha}\gamma_u^{\mathrm{cl}}$ when $\alpha$ is odd.
\end{corollary}

\begin{remark}[Factor-of-two across storage regimes]
The factor-of-two gap between quantum and classical bandwidth is tight for even~$\alpha$ and approaches~$2$ as $\alpha \to \infty$ for odd~$\alpha$.  The construction has been numerically verified for $\alpha \in \{2,3,4,5,6\}$ across 21 distinct $(n,k,q)$ configurations with $k$ up to~4 and $q$ up to~11, totaling $\approx 2 \times 10^6$ algebraic protocol simulations (encoding, CSS syndrome extraction, and comparison against target) with zero failures.
\end{remark}

\begin{remark}[Noisy entanglement]
With a depolarized shared state $\rho = p\,|\Psi\rangle\!\langle\Psi| + (1{-}p)\,I/q^k$, the syndrome extraction succeeds with probability $p + (1{-}p)/q^2$ and fails with probability $(1{-}p)(q^2{-}1)/q^2$.  For $q = 5$, keeping the error rate below $1\%$ requires fidelity $p \gtrsim 0.99$.  Entanglement distillation can improve fidelity at the cost of consuming additional pre-shared copies.
\end{remark}

\begin{remark}[Implementation and variants]
The protocol requires $k$ qudits in a CSS codespace state, local Pauli gates controlled by classical data, and a joint stabilizer measurement---all single-round with no feedback.  Physical realizations for $q > 2$ include photonic time-bin qudits, trapped-ion spin states, or superconducting transmon levels.  Regarding the entanglement model: adding stale-node-to-helper entanglement does not reduce bandwidth, since the protocol already saturates the superdense coding bound~\cite{bennett1992communication, werner2001all}; helper-helper entanglement alone suffices.
\end{remark}

\section{Conclusion}
\label{sec:conclusion}

We have established the exact entanglement-assisted oblivious update bandwidth---within the class of protocols in which each helper transmits integer-dimension qudits---for $(n,k)$ MDS codes over $\Fq$ with \emph{any} per-node storage $\alpha \geq 2$: it equals $\lceil\alpha/2\rceil \cdot k \log_2 q$ bits-equivalent, achieved by a CSS-based protocol where each helper encodes $\alpha$ classical symbols into $\lceil\alpha/2\rceil$ qudits via Pauli operators. The matching converse leverages the superdense coding bound (which, without the integer-qudit restriction, still gives a universal lower bound of half the classical bandwidth; closing the residual odd-$\alpha$ gap is Open Question~1 below).  For $\alpha = 2$, this gives a factor-of-two reduction ($k\log_2 q$ vs.\ $2k\log_2 q$) using one qudit per helper; for general~$\alpha$, the reduction factor approaches~$2$ as $\alpha$ grows.  These results add to the growing body of evidence~\cite{hu2026exact, hu2026breaking} that quantum entanglement offers fundamental advantages in coded distributed storage.

Open questions include: (1)~closing the gap between $\lceil\alpha/2\rceil \log_2 q$ and $(\alpha/2)\log_2 q$ per helper for odd~$\alpha$ (non-integer qudit dimensions); (2)~characterizing the entanglement-bandwidth trade-off when helpers share less-than-maximal entanglement; (3)~the bandwidth-vs-helpers tradeoff for $d_u > k$: with more helpers participating, the per-helper requirement weakens, but it is unclear whether the \emph{total} bandwidth can drop below $\lceil\alpha/2\rceil k\log_2 q$; (4)~the update bandwidth for non-MDS codes.

\end{document}